\documentclass[a4paper]{article}
\usepackage{graphicx} 
\usepackage{amsmath, amsthm, amsfonts, amssymb, amscd, bm, enumerate}
\usepackage{verbatim}
\usepackage{url}

\newtheorem{theorem}{Theorem}
\newtheorem{corollary}{Corollary}
\newtheorem{lemma}{Lemma}
\newtheorem{proposition}{Proposition}
\newtheorem{definition}{Definition}

\theoremstyle{remark}
\newtheorem{remark}{Remark}
\newtheorem{example}{Example}
\def\F{\mathcal{F}}
\def\G{\mathcal{G}}
\def\CE{\mathcal{E}}
\def\oCE{\overleftarrow{\mathcal{E}}\hspace{-0.2em}}
\def\CF{\mathcal{F}}
\def\CM{\mathcal{M}}
\def\CA{\mathcal{A}}
\def\CR{\mathcal{R}}
\def\fA{\mathfrak{A}}
\def\ofA{\overleftarrow{\mathfrak{A}}}

\def\bA{\mathbb{A}}
\def\bR{\mathbb{R}}
\def\bN{\mathbb{N}}
\def\bQ{\mathbb{Q}}
\def\bP{\mathbb{P}}
\def\bE{\mathbb{E}}

\def\CY{\mathcal{Y}}

\def\bone{\mathbf{1}}

\begin{document}

\title{Uncertainty and filtering of hidden Markov models in discrete time}
\author{Samuel N. Cohen\footnote{Research supported by the Oxford--Man Institute for Quantitative Finance and the Oxford--Nie Financial Data Laboratory. Thanks to Ramon van Handel, Michael Monoyios, Sergey Nadtochiy, Andrew Allan, Gon{\c c}alo Sim{\~ o}es and Robert Elliott for useful conversations.}\\samuel.cohen@maths.ox.ac.uk\\ Mathematical Institute, University of Oxford}
\date{\today}

\maketitle
\begin{abstract}
We consider the problem of filtering an unseen Markov chain from noisy observations, in the presence of uncertainty regarding the parameters of the processes involved. Using the theory of nonlinear expectations, we describe the uncertainty in terms of a penalty function, which can be propagated forward in time in the place of the filter. We also investigate a simple control problem in this context.

\end{abstract}

\section{Introduction}

Filtering is a common problem in many applications. The essential concept is that there is an unseen Markov process, which influences the state of some observed process, and our task is to approximate the state of the unseen process using a form of Bayes' theorem. Many results have been obtained in this direction, most famously the Kalman filter (Kalman \cite{52}, Kalman and Bucy \cite{53}), which assumes the underlying processes considered are Gaussian, and gives explicit formulae accordingly. Similarly, under the assumption that the underlying process is a finite-state Markov chain, a general formula to calculate the filter can be obtained (the Wonham filter \cite{79}). These results are well known, in both discrete and continuous time (see Bain and Crisan \cite{Bain2009} or Cohen and Elliott \cite[Chapter 21]{Cohen2015} for further general discussion).

In this paper, we shall consider a simple setting in discrete time, where the underlying process is a finite-state Markov chain. Our concern will be to study uncertainty in the dynamics of the underlying processes, in particular its effect on the behaviour of the corresponding filter. That is, we will assume that the observer has only imperfect knowledge of the dynamics of the underlying process and of its relationship with the observation process, and wishes to incorporate this uncertainty in their estimates of the unseen state. We are particularly interested in allowing the level of uncertainty in the filtered state to be endogenous to the filtering problem, arising from the uncertainty in parameter estimates and process dynamics.

We will model this uncertainty in a general manner, using the theory of nonlinear expectations, and shall particularly concern ourselves with a description of uncertainty for which explicit calculations can still be carried out, and which can be motivated by considering statistical estimation of parameters. We then apply this to building a dynamically consistent expectation for random variables based on future states, and to a general control problem with learning under uncertainty.

\subsection{Basic filtering}\label{sec:BasicFilter}

Consider two stochastic processes, $X=\{X_t\}_{t\ge 0}$ and $Y=\{Y_t\}_{t\ge 0}$. Let $\Omega$ be the space of paths of $(X,Y)$ and $\bP$ be a probability measure on $\Omega$. We denote by $\{\CF_t\}_{t\ge 0}$ the (completed) filtration generated by $X$ and $Y$, and $\CY=\{\CY_t\}_{t\ge 0}$ the (completed) filtration generated by $Y$. The key problem of filtering is to determine estimates of $\phi(X_t)$ given $\CY_t$, that is $\bE_\bP[\phi(X_t)|\CY_t]$ where $\phi$ is an arbitrary Borel function. 

Suppose that $X$ is a Markov chain with (possibly time-dependent) transition matrix $A_t^\top$ under $\bP$ (the transpose here saves notational complexity later).  Without loss of generality we can assume that $X$ takes values in the standard basis vectors $\{e_i\}_{i=1}^N$ of $\bR^N$ (where $N$ is the number of states of $X$), and so we can write 
\[X_t = A_t X_{t-1} + M_t\]
where $\bE_\bP[M_{t+1}|\CF_t] = 0$, so $\bE_\bP[X_t|\F_{t-1}] = A_t X_{t-1}$.

We suppose the process $Y$ is multivariate real-valued\footnote{This assumption can easily be relaxed, to allow for $Y$ to take values in an appropriate Polish or Blackwell space. We restrict to the real setting purely for simplicity.}. The law of $Y$ will be allowed to depend on $X$, in particular, the $\bP$-distribution of $Y_t$ given $\{X_{s}\}_{s\leq t}\cup \{Y_{s}\}_{s< t}$ (that is, given all past observations of $X$ and $Y$ and the current state of $X$) is 
\[Y_t \sim c(y;t, X_t)d\mu(y)\]
for $\mu$ a reference measure on $(\bR^d, \mathcal{B}(\bR^d))$.

For simplicity, we shall assume that $Y_0\equiv 0$, so no information is revealed about $X_0$ at time $0$. It is convenient to write $C_t(y)=C(y; t)$ for the diagonal matrix with entries $c(y;t, e_i)$, so that \[ C_t(y) X_t = c(y;t, X_t)X_t.\]  Note that these assumptions, in particular the values of $A$ and $c$, depend on the choice of probability measure $\bP$. Conversely, as our space $\Omega$ is the space of paths of $(X,Y)$, the measure $\bP$ is determined by $A$ and $c$.

As we have assumed $X_t$ takes values in the standard basis in $\bR^N$, the expectation $\bE_\bP[X_{t}|\CY_{t}]$ determines the entire conditional distribution of $X_t$ given $\CY_t$. In this discrete time context, the filtering problem can be solved in a fairly simple manner: Suppose we have already calculated $p_{t-1}:=\bE_\bP[X_{t-1}|\CY_{t-1}]$. Then by linearity and the dynamics of $X$, using the fact
\[\bE_\bP[M_t|\CY_{t-1}] = \bE_\bP[\bE_\bP[M_t|\F_{t-1}]|\CY_{t-1}]=0,\]
we can calculate
\[\bE_\bP[X_{t}|\CY_{t-1}] = \bE_\bP[A_t X_{t-1}+ M_t|\CY_{t-1}] =  A_t p_{t-1}.\]

Bayes' theorem then states that, with probability one,
\[\bP(X_t=e_i|\CY_{t}) = \bP(X_t=e_i|\{Y_{s}\}_{s< t}, Y_t) \propto c(Y_t;t, e_i) \bP(X_t=e_i|\CY_{t-1}),\]
which can be written in a simple matrix form,
\begin{equation}\label{eq:basicfilter}
 p_t \propto C_t(Y_t) A_t p_{t-1}.
\end{equation}
As $p_t$ is a probability vector, normalization of the right hand side determines $p_t$ directly. We call $p_t$ the `filter state' at time $t$. Note that, if we assume the density $c$ is positive, $A_t$ is irreducible and $p_{t-1}$ has all entries positive, then $p_t$ will also have all entries positive.

\begin{definition}
For future use, if $\fA=(A,C(\cdot))$ denotes the $A$ and $C$ matrices described above and $p_0$ is the initial filter state, then we will write $p_t^{\fA,p_0}$ for the filter state at time $t$, that is, the solution to \eqref{eq:basicfilter} (where the observations $Y$ are implicit).
\end{definition}

 In practice, the key problem with implementing these methods is the requirement that we know the underlying transition matrix $A^\top$ and the density $c$. These are generally not known perfectly, but need to be estimated prior to the implementation of the filter. Uncertainty in the choice of these parameters will lead to uncertainty in the estimates of the filtered state, and the aim of this paper is to derive useful representations of that uncertainty. 

As variation in the choice of $A$ and $c$ corresponds to a different choice of measure $\bP$, we see that using an uncertain collection of generators corresponds naturally to uncertainty regarding $\bP$. This type of uncertainty, where the probability measure is not known, is commonly referred to as `Knightian' uncertainty (with reference to Knight \cite{Knight1921}, related ideas are also discussed by Keynes \cite{Keynes1921}). 

Effectively, we wish to consider the propagation of uncertainty in Bayesian updating (as the filter is simply a special case of this). Huber and Ronchetti \cite[p331]{Huber2009} briefly touch on this, however (based on earlier work by Kong) argue that this propagation is computationally infeasible. However, their approach was based on Choquet integrals, rather than nonlinear expectations. In the coming sections, we shall see how the structure of nonlinear expectations allows us to derive comparatively simple rules for updating.

\begin{remark}
While we will present our theory in the context where $X$ is a finite state Markov chain, our approach does not depend in any significant way on this assumption. In particular, it would be equally valid \emph{mutatis mutandis} when we supposed that $X$ followed the dynamics of the Kalman filter, and our uncertainty was on the coefficients of the filter. We specialize to the Markov chain case purely for the sake of concreteness.
\end{remark}

\section{Conditionally Markov Measures}
In order to incorporate learning in our nonlinear expectations and filtering, it is useful to extend slightly from the family of measures previously described. In particular, we wish to allow the dynamics to depend on the past observations, while preserving enough Markov structure to enable filtering. We write $\CM_1$ for the space of probability measures equivalent to a reference measure $\bP$. The following two classes of probability measures will be of interest.

\begin{definition}
 Let $\CM_M\subset \CM_1$ denote the probability measures under which 
\begin{itemize}
 \item $X$ is a Markov chain, that is, for all $t$, $X_{t+1}$ is independent of $\F_t$ given $X_t$,
 \item $\{Y_s\}_{s\ge {t+1}}$ is independent of $\F_t$ given $X_{t+1}$,
 \item both $X$ and $Y$ are time homogeneous, that is, the conditional distributions of $X_{t+1}|X_{t}$ and $Y_t|X_t$ do not depend on $t$.
\end{itemize}
\end{definition}

To extend this slightly, we can allow our processes to depend on the past of the observed process $Y$.
\begin{definition}
 Let $\CM_{M|\CY}\subset \CM_1$ denote the probability measures under which 
\begin{itemize}
 \item $X$ is a conditional Markov chain, that is, for all $t$, $X_{t+1}$ is independent of $\F_t$ given $X_t$ and $\{Y_s\}_{s\leq t}$, and
 \item $\{Y_s\}_{s\ge {t+1}}$ is independent of $\F_t$ given $\{X_{t+1}\}\cup\{Y_s\}_{s\leq t}$.
\end{itemize}
\end{definition}

We should note that, if we consider a measure in $\CM_{M|\CY}$, there is a natural notion of the generators $A$ and $C$. In particular, $\CM_M$ corresponds to those measures under which the generators $A$ and $C$ are constant, while $\CM_{M|\CY}$ corresponds to those measures under which the generators $A$ and $C$ are deterministic functions of time and $\{Y_s\}_{s\leq t}$. 

\begin{definition}We shall write $\mathbb{A}$ for the space in which the generator takes values\footnote{This space can be thought of as the product of the space of transition matrices and the space of diagonal matrix-valued functions, where each diagonal element is a probability density on $\mathbb{R}^d$.} and write $\mathcal{A}_{\CY}$ for the collection of generators associated with $\CM_{M|\CY}$, that is, $\CY$-adapted processes taking values in $\mathbb{A}$. 
\end{definition}

For each $t$, these generators determine the measure on $\F_t$ given $\F_{t-1}$, and (together with the distribution of $X_0$) this determines the measure at all times. It is straightforward to verify that our filtering equations hold for all measures in $\CM_{M|\CY}$, with the appropriate modification of the generators. 

\begin{definition}
For a measure $\bQ\in \CM_{M|\CY}$, we shall write $\fA^\bQ = \big(A^\bQ,C^\bQ(\cdot)\big)$ for the generator of $(X,Y)$ under $\bQ$, recalling that $C^\bQ_t(y) = \mathrm{diag}(\{c^\bQ_t(y; e_i)\}_{i=1}^N)$, and that $A^\bQ_t$ and $C^\bQ_t$ are now allowed to depend on $\{Y_s\}_{s<t}$.  For notational convenience, we shall typically not write the dependence on $\{Y_s\}_{s<t}$ explicitly.

Similarly, for $\fA\in \CA_{\CY}$ and $p_0$ a probability vector in $\bR^N$, we shall write $\bQ^{\fA, p_0}$ for the measure with generator $\fA$ and $X_0\sim p_0$ under $\bQ$.
\end{definition}

In our setting, our fundamental problem is that we do not know what measure is `true', and so work instead under a family of measures. In general, measure changes can be described as follows.

\begin{proposition}\label{prop:LikeFull}
Let $\bar\bP$ be the reference measure (which does not have to be a probability measure), under which $X$ is a sequence of iid uniform random variables from the basis vectors $\{e_1,...e_N\}\subset\bR^N$ and $\{Y_t\}_{t\ge 0}$ is independent of $X$, with iid distribution $Y_t\sim d\mu$. The measure $\mathbb{Q}\in \CM_{M|\CY}$ where $X$ has generator $\{A^\mathbb{Q}_t\}_{t\ge 0}$, where $Y_t\sim c^\mathbb{Q}_t(y, X_t)d\mu(y)$ and $X_0\sim p^\mathbb{Q}_0$ has Radon--Nikodym derivative (or likelihood)
\[
\frac{d\mathbb{Q}}{d\bar\bP}\Big|_{\F_T} = (X_0^\top p_0^\bQ) N\prod_{t=1}^T \Big(\big(X_t^\top A_{t-1}^{\bQ} X_{t-1}\big)c^\bQ_t(Y_t; X_t)\Big).
\]
\end{proposition}

The above proposition gives a Radon--Nikodym derivative adapted to the full filtration $\{\F_t\}_{t\ge 0}$. In practice, it is also useful to consider the corresponding Radon--Nikodym derivative adapted to the observation filtration $\{\CY_t\}_{t\ge 0}$. As this filtration is generated by the process $Y$, it is enough to multiply together the conditional distributions of $Y_t|\CY_{t-1}$, leading to the following convenient representation. For notational simplicity, we write 
\[c_t(y; p):=\sum_i p_i c_t(y; e_i).\]

\begin{proposition}\label{prop:LikePart}
For $\mathbb{Q}$ as in Proposition \ref{prop:LikePart}, the Radon--Nikodym derivative restricted to $\CY_T$ is given by 
\[
\frac{d\mathbb{Q}}{d{\bar\bP}}\Big|_{\CY_T} = \prod_{t=1}^T c^\bQ_t(Y_t; A^\bQ p_{t-1}^{\fA^\bQ, p^\bQ_0})
\]
where we recall that $p_t^{\fA^\bQ, p^\bQ_0}$ is the solution to the filtering problem in the measure $\bQ\in \CM_{M|\CY}$, as determined by \eqref{eq:basicfilter} (and so includes further dependence on $\{Y_{s}\}_{s <t}$).
\end{proposition}

\section{Nonlinear Expectations}
 In this section we introduce the concepts of nonlinear expectations and convex risk measures, and discuss their connection with penalty functions on the space of measures. These objects provide a technical foundation with which to model the presence of uncertainty in a random setting. This theory is explored in some detail in F\"ollmer and Schied \cite{Follmer2002}. Other key works which have used or contributed to this theory, in no particular order, are Hansen and Sargent \cite{Hansen2008} (see also \cite{Hansen2005, Hansen2007} for work related to what we present here), Huber and Ronchetti \cite{Huber2009}, Peng \cite{Peng2010}, El Karoui, Peng and Quenez \cite{El1997}, Delbaen, Peng and Rosazza Gianin \cite{Delbaen2008}, Duffie and Epstein \cite{Duffie1992}, Rockafellar, Uryasev and Zabarankin \cite{Rockafellar2006},  Riedel \cite{Riedel2004} and Epstein and Schneider \cite{Epstein2003}. We base our terminology on that used in \cite{Follmer2002} and \cite{Delbaen2008}.

 We here present, without proof, the key details of this theory as needed for our analysis.  
 
 \begin{definition}\label{defn:nonlinearexpectation}
  For a $\sigma$-algebra $\G$ on $\Omega$,  let $L^\infty(\G)$ denote the space of essentially bounded $\G$-measurable random variables. A nonlinear expectation on $L^\infty(\G)$ is a mapping 
 \[\CE:L^\infty(\G) \to \bR\]
 satisfying the assumptions
 \begin{itemize}
  \item Strict Monotonicity: for any $\xi_1, \xi_2\in L^\infty(\G)$, if $\xi_1\geq \xi_2$ a.s. then $\CE(\xi_1) \geq \CE(\xi_2)$, and if in addition $\CE(\xi_1)=\CE(\xi_2)$ then $\xi_1=\xi_2$ a.s.,
  \item Constant triviality: for any constant $k$, $\CE(k)=k$,
  \item Translation equivariance: for any $k\in\bR$, $\xi\in L^\infty(\G)$, $\CE(\xi+k)= \CE(\xi)+k$.
 \end{itemize}
 A `convex' expectation in addition satisfies
 \begin{itemize}
  \item Convexity: for any $\lambda\in [0,1]$, $\xi_1, \xi_2\in L^\infty(\G)$, 
 \[\CE(\lambda \xi_1+ (1-\lambda) \xi_2) \leq \lambda \CE(\xi_1)+ (1-\lambda) \CE(\xi_2).\]
 \end{itemize}
 
 If $\CE$ is a convex expectation, then the operator defined by $\rho(\xi) = \CE(-\xi)$ is called a \emph{convex risk measure}. A particularly nice class of convex expectations is those which satisfy
 \begin{itemize}
  \item Lower semicontinuity: For a sequence $\{\xi_n \}_{n\in\bN}$ with $\xi_n \uparrow \xi$ pointwise (and $\xi\in L^\infty(\G)$), $\CE(\xi_n) \uparrow \CE(\xi)$.
 \end{itemize}
 \end{definition}
 
 The following theorem (which was expressed in the language of risk measures) is due to F\"ollmer and Schied \cite{Follmer2002a} and Frittelli and Rosazza Gianin \cite{Frittelli2002}.
 \begin{theorem}\label{thm:penaltyexists}
Suppose $\CE$ is a lower semicontinuous convex expectation. Then there exists a  `penalty' function $\CR: \CM_1\to [0,\infty]$ such that 
 \[\CE(\xi) = \sup_{\bQ\in \CM_1} \big\{\bE_\bQ[\xi] -\CR(\bQ)\big\}.\]
Provided $\CR(\bQ)<\infty$ for some $\bQ$ equivalent to $\bP$, we can restrict our attention to measures in $\CM_1$ equivalent to $\bP$ without loss of generality.
 \end{theorem}
 
 \begin{remark}
  This result gives some intuition as to how a convex expectation can model `Knightian' uncertainty. One considers all the possible probability measures on the space, and then selects the maximal expectation among all measures, penalizing each measure depending on how plausible it is considered. As convexity of $\CE$ is a natural requirement of an `uncertainty averse' assessment of outcomes, Theorem \ref{thm:penaltyexists} shows that this is the only way to construct an `expectation' $\CE$ which penalizes uncertainty, while preserving monotonicity, translation equivariance and constant triviality.
\end{remark}
 
\subsection{DR-expectations}
From the discussion above, it is apparent that we can focus our attention on calculating the penalty function $\CR$, rather than the nonlinear expectation directly. This penalty function is meant to encode how `unreasonable' a probability measure $\bQ$ is as a model for our outcomes. So far, we have assumed that the penalty did not depend on time or on observations. By relaxing this assumption, we can incorporate learning of which models are `good' in our framework. 

In \cite{Cohen2016}, we have considered a framework which links the choice of the penalty function to statistical estimation of a model. The key idea of \cite{Cohen2016} is to use the negative log-likelihood function for this purpose, where the likelihood is taken against an arbitrary reference measure, and evaluated using the observed data. This directly uses the statistical information from observations in the quantification of uncertainty. 

In this paper, we shall make a slight extension of this idea, to explicitly incorporate prior beliefs. In particular, we shall replace the log-likelihood with the log-posterior density, which in turn gives an additional term in the penalty. In order to be precise, we now give a formal definition of the likelihood, which is sufficient for our purposes.

\begin{remark}
 In what follows, we will be variously wishing to \emph{restrict} a measure $\mathbb{Q}$ to a $\sigma$-algebra, and to \emph{condition} a measure on a $\sigma$-algebra. To prevent notational confusion, we shall write $\mathbb{Q}\|_\mathcal{F}$ for the restriction of $\mathbb{Q}$ to $\mathcal{F}$, and $\mathbb{Q}|_\mathcal{F}$ for $\mathbb{Q}$ conditioned on $\mathcal{F}$.
\end{remark}

\begin{definition}

% For a model $\bQ\in \CM_1$, let $L(\bQ|\mathbf{y})$ denote the likelihood of a random variable $\mathbf{y}$ under $\bQ$, that is the density $d\bQ/d\bar\bP\|_{\sigma(\mathbf{y})}$ with respect to a reference measure $\bar\bP$ evaluated at $\mathbf{y}$. If a prior term is included, then let $L$ denote the posterior density given $\mathbf{y}$, rather than the likelihood.

Let $\mathcal{Q}\subseteq\CM_1$ be a set of models under consideration (for example, a parametric set of distributions). For observations $\mathbf{y}$ taking values in $\bR^N$, we define the likelihood to be a fixed map $L^{\mathrm{obs}}:\mathcal{Q}\times \bR^N \to \bR$, measurable with respect to its second argument, such that  $\omega \mapsto L^\mathrm{obs}(\bQ|\mathbf{y}(\omega))$ is a version of the Radon--Nikodym derivative $d\bQ\|_{\sigma(\mathbf{y})} / d\bar\bP\|_{\sigma(\mathbf{y})}$.

Inspired by a `Bayesian' approach, we augment this by the addition of a prior distribution over $\mathcal{Q}$. Suppose a (possibly improper\footnote{As we have not specified a reference measure over $\mathcal{Q}$, we have not defined the prior density as a Radon--Nikodym derivative, and cannot integrate it over the class of models. Therefore, we do not require it to `integrate to $1$', that is, we have an improper prior.  This has no significant impact in what follows, as \eqref{eq:divergence} normalizes away the effect of the reference distribution.}) prior with density of the form $\exp(-\gamma(\bQ))$ is given, then we define the posterior relative density
\[L(\bQ|\mathbf{y}) = L^{\mathrm{obs}}(\bQ|\mathbf{y})\exp(-\gamma(\bQ)).\]

  We then define the ``$\mathcal{Q}|\mathbf{y}$-divergence'' to be the negative log-likelihood ratio (or log-posterior relative density)
\begin{equation}\label{eq:divergence}\alpha_{\mathcal{Q}|\mathbf{y}}(\bQ):= -\log\big(L(\bQ|\mathbf{y})\big) + \sup_{\tilde \bQ\in \mathcal{Q}}\Big\{\log\big(L(\tilde \bQ|\mathbf{y})\big)\Big\}.
\end{equation}
\end{definition}
\begin{remark}
 The right hand side of \eqref{eq:divergence} is well defined whether or not a maximum a posteriori estimator\footnote{Recall that a $\mathcal{Q}$-MAP (maximum a posteriori estimator) is a map $\mathbf{y}\to \hat \bQ\in \mathcal{Q}$ such that $L(\hat \bQ|\mathbf{y}) \geq L(\bQ|\mathbf{y})$ for all $\bQ\in\mathcal{Q}$. }  exists. Given a $\mathcal{Q}$-MAP $\hat \bQ$, we would have the simpler representation 
\[\alpha_{\mathcal{Q}|\mathbf{y}}(\bQ):= -\log\Big(\frac{L(\bQ|\mathbf{y})}{L(\hat \bQ|\mathbf{y})}\Big).\]
\end{remark}

\begin{definition}\label{defn:DRexp}
For fixed observations $\mathbf{y}_t=(Y_1,Y_2,..., Y_t)$, for an uncertainty aversion parameter $k>0$ and exponent $k' \in[1,\infty]$, we define the convex expectation
\begin{equation}\label{eq:DRexpDefn}
 \mathcal{E}_{\mathcal{Q}|\mathbf{y}_t}^{k,k'}(\xi):= \sup_{\bQ\in \mathcal{Q}}\Big\{\bE_\bQ[\xi|\mathbf{y}_t] -\Big(\frac{1}{k}\alpha_{\mathcal{Q}|\mathbf{y}_t}(\bQ)\Big)^{k'}\Big\},
\end{equation}
where we adopt the convention $x^\infty = 0$ for $x\in[0,1]$ and $+\infty$ otherwise.

We call $\mathcal{E}_{\mathcal{Q}|\mathbf{y}_t}^{k,k'}$ the ``DR-expectation\footnote{DR refers either to divergence-robust or data-driven robust.}'' (with parameter $k,k'$). We may omit to write $k, k'$ for notational simplicity. 
\end{definition}

\begin{remark}
 In the cases of interest for this paper, we shall assume that $\mathcal{Q}$ is parameterized by some finite-dimensional real value, such that the divergence and conditional expectations given $\mathbf{y}_t$ are continuous with respect to this parameterization, and are Borel measurable with respect to $\mathbf{y}_t$. This means that the measure theoretic concerns which arise from our definitions in terms of the Radon--Nikodym derivative and taking the supremum will not cause difficulty, in particular, the DR-expectation defined in \eqref{eq:DRexpDefn} is guaranteed to be a Borel measurable function of $\mathbf{y}_t$ for every $\xi$. (This follows from Filippov's implicit function theorem.)
\end{remark}

\begin{remark}\label{rem:basicDRapplication}
 In theory, we could now apply the DR-expectation framework to a filtering context as follows: Take a collection of models $\mathcal{Q} \subseteq \CM_{M|\CY}$. For a random variable $\xi$, and for each measure $\mathbb{Q}\in\mathcal{Q}$, compute $E_\mathbb{Q}[\xi|\mathbf{y}_t]$ and $\alpha_{\mathcal{Q}|\mathbf{y}_t}(\bQ)$. Taking a supremum as in \eqref{eq:DRexpDefn}, we obtain the DR-expectation. However, this is generally not computationally tractable in this form.
\end{remark}

\begin{lemma}\label{lem:horizonirrelevant}
Let $\{\CF_t\}_{t\ge 0}$ be a filtration such that $Y$ is adapted. For $\CF_t$-measurable random variables, the choice of horizon $T\geq t$ in the definition of the penalty function $\alpha$ is irrelevant. That is, for $\CF_t$-measurable $\xi$ and any $s\ge t$,
\[ \mathcal{E}_{\mathcal{Q}|\mathbf{y}_t}(\xi) = \sup_{\bQ\in \mathcal{Q}}\Big\{\bE_\bQ[\xi|\mathbf{y}_t] -\Big(\frac{1}{k}\alpha_{\mathcal{Q}|\mathbf{y}_t}(\bQ\|_{\F_s})\Big)^{k'}\Big\},
\]
where $\alpha_{\mathcal{Q}|\mathbf{y}_t}(\bQ\|_{\F_s})$ is defined as above, in terms of the restricted measure $\bQ\|_{\F_s}$.
\end{lemma}
\begin{proof}
 By definition, the likelihood is determined by the restriction of $\bQ$ to $\CY_t \subset \F_t \subset \F_s$, while the expectation depends only on the restriction of $\bQ$ to $\CF_t\subset \F_s$. As these are the only terms needed to compute the DR-expectation, the result follows.
\end{proof}

\begin{remark}
The purpose of the nonlinear expectation is to give an `upper' estimate of a random variable, accounting for uncertainty in the underlying probabilities. This is closely related to robust estimation in the sense of Wald \cite{Wald1945}. In particular,  one can consider the robust estimator given by 
\[\mathrm{arginf}_{\hat\xi\in\bR^N} \CE(\|\xi -\hat \xi\|^2|\CY_t),\]
which gives a `minimax' estimate of $\xi$, given the observations $\CY_t$ and a quadratic loss function. The advantage of the nonlinear expectation approach is that it allows one to construct such an estimate for \emph{every} random variable/loss function, giving a cost-specific quantification of uncertainty in each case. 

We can also see a connection with the theory of $H^\infty$ filtering (see, for example  Grimble and El Sayed \cite{Grimble1990} or more recently Zhang, Xia and Shi \cite{Zhang2009} and references therein, or the more general $H^\infty$-control theory in Ba\c{s}ar and Bernhard \cite{Basar1991}). In this setting, we look for estimates which perform best in the worst-situation, where `worst' is usually defined in terms of a perturbation to the input signal or coefficients. In our setting, we focus not on the estimation problem directly, but on the `dual' problem of building an upper expectation, i.e. calculating the `worst' expectation in terms of a class of perturbations to the coefficients (our setting is general enough that perturbation to the signal can also be included, through shifting the coefficients). 
\end{remark}

\begin{remark}
There are also connections between our approach and what is called `risk-sensitive filtering', see for example James, Baras and Elliott \cite{James1994},  Dey and Moore \cite{Dey1995}, or the review of Boel, James and Petersen \cite{Boel2002} and references therein (from an engineering perspective) or Hansen and Sargent \cite{Hansen2007, Hansen2008} (from an economic perspective). In their setting, one uses the nonlinear expectation defined by 
\[\CE(\xi|\CY_t) = -k \log \bE_\bP\big[\exp(- \xi/k)\big|\CY_t\big],\]
for some choice of robustness parameter $1/k>0$. This leads to significant simplification, as dynamic consistency and recursivity is guaranteed in every filtration (see Graf \cite{Graf1980} and Kupper and Schachermeyer \cite{Kupper2008}, and further discussion in Section \ref{sec:consist}) and the corresponding penalty function is given by the conditional relative entropy
\[\CR_t(\bQ) = k \bE_\bQ[\log(d\bQ/d\bP)|\CY_t],\]
the one-step penalty can be calculated accordingly.  The optimization for the nonlinear expectation can be taken over $\CM_1$, so this approach has a claim to be including `nonparametric' uncertainty, as all measures are considered, rather than purely Markov measures or measures in a parametric family (however the optimization can be taken over conditionally Markov measures, and one will obtain an identical result!). 

The difficulty with this approach is that it does not allow for easy incorporation of knowledge of the error of estimation of the generator $\fA$ in the level of robustness -- the only parameter available to choose is $k$, which multiplies the relative entropy. A small choice of $k$ corresponds to a small penalty, hence a very robust expectation, but this robustness is not directly linked to the estimation of the generators $\fA$. Therefore the impact of statistical estimation error remains obscure, as $k$ is chosen largely exogenously of this error. For this reason, our approach, which directly allows for the penalty to be based on the statistical estimation of the generators, has advantages over this simpler method.
\end{remark}

\section{Recursive penalties}

The DR-expectation provides us with an approach to including statistical estimation in our valuations. However, the calculations suggested by Remark \ref{rem:basicDRapplication} are generally intractable in their stated form. In this section, we shall see how the DR-expectation approach, and an approach with a constant penalty $\CR$, specialize in a filtering setting.

The class of models we shall consider are based on two key questions:
\begin{enumerate}
\item \emph{``Static or Dynamic generators?''} Are the generators $A$ and $C$ static (through time) and unknown (so can be estimated) or are they not only unknown but dynamically changing (and we depend principally on prior information about their likely behaviour)? 

In other words, do our models come from $\CM_M$ (static generators) or from $\CM_{M|\CY}$ (dynamic generators)?

\item \emph{``Uncertain prior (UP) or DR-expectations?''} Do we (i) have a fixed penalty on the `reasonableness' of a model, and then use new observations to update \emph{within} each model using Bayes' theorem, or are we (ii) attempting to determine which model is reasonable \emph{using} our observations, while simultaneously updating our model (with the same observations). 

In other words, is our penalty $\CR$ constant (UP framework) or does it change with new observations (DR framework)?
\end{enumerate}
% The answer to the first question fundamentally determines the dimension of the state space for our penalty function, while the second controls the choice of recursive dynamics.

For practical purposes, it is critical that we refine our approach to provide a recursive construction of our nonlinear expectation. In classical filtering, one obtains a recursion for expectations $\bE[\phi(X_t)|\CY_t]$, for Borel functions $\phi$; one does not typically consider the expectations of general random variables. In the same way, we will consider the expectations of random variables $\phi(X_t)$.

It is clear that we can consider $\CE_{\mathcal{Q}|\mathbf{y}_t}(\phi(X_t))$ as a nonlinear expectation with the probability space being only the states of $X_t$. By Theorem \ref{thm:penaltyexists}, it follows that, for each $t$, there exists a $\CY_t\otimes\mathcal{B}(\bR)$-measurable function $\kappa_t$ such that
\begin{equation}\label{eq:recursivePen}
 \begin{split}
\CE_{\mathcal{Q}|\mathbf{y}_t}(\phi(X_t)) &:= \sup_{\bQ\in\mathcal{Q}} \Big\{\bE_\bQ[\phi(X_t)|\mathbf{y}_t]-\CR_t(\bQ)\Big\}\\ &= \sup_{q\in S_N^+}\Big\{\sum_i q_i \phi(e_i) - \Big(\frac{1}{k}\kappa_t(\omega, q)\Big)^{k'}\Big\},\end{split}
\end{equation}
where $S_N^+$ denotes the probability simplex in $\bR^N$, that is, $S_N^+=\{x\in \bR^N: \sum_i x_i = 1,\quad  x_i\geq 0 \quad\forall i\}$. Our aim is to find a recursion for $\kappa_t$, for various choices of $\CR$. Without loss of generality, we will write 
\[\CR_t(\bQ) =  \Big(\frac{1}{k}\alpha_{(...)}(\bQ)\Big)^{k'},\]
where $\alpha$ has an appropriate set of arguments, as this gives consistent notation between the DR-expectation and Uncertain Prior settings.

If the generators are assumed to be static but unknown, the following definition will prove useful.
\begin{definition}
Let $K_t:\Omega\times S_N^+ \times \bA \to \bR$ denote an extended penalty function, which encodes the penalty associated with a state $X_t|\CY_t\sim p_t\in S_N^+$ and a generator $\mathfrak{A}\in\bA$. The penalty $\kappa_t:\Omega\times S_N^+\to\bR$ is obtained from $K_t$ through the relation
\begin{equation}\label{eq:kappaKDefn}
\kappa_t(\omega, p) = \inf_{\mathfrak{A}\in \bA} K_t(\omega, p,\mathfrak{A}).
 \end{equation}
\end{definition}

The reason for this definition is that, in a static generator framework, it is $K$, not $\kappa$, which will satisfy a recursive equation. We note, however, that the dimension of $(p, \mathfrak{A})$ is larger, and  typically much larger, than $N-1 = \mathrm{dim}S_N^+$, which leads to difficulty in implementation.

Our constructions will depend on the following object.
\begin{definition}\label{defn:backprojection}
For each generator $\fA=(A,C)\in\bA$ and each $p\in S_N^+$, we define the random set 
\[(p)_t^{\ofA} = \big\{p_{t-1}\in S_N^+: C(Y_t)Ap_{t-1} \propto p\big\}.\]
By recursion, we extend this to all $s<t$
\[(p)_t^{(\ofA,s-1)} = \big\{p_{s-1}\in S_N^+: C(Y_s)Ap_{s-1} \propto p_s\in (p)_t^{(\ofA,s)} \big\}.\]
\end{definition}
The set $(p)_t^{\ofA}$ represents the filter states at time $t-1$ which evolve to $p$ at time $t$. Clearly this is only known at time $t$, as it depends on $Y_t$. Similarly, we think of $(p)_t^{(\ofA,s)}$ as the set of filter states at time $s$ which would evolve to the state $p$ at time $t$, given the generator $\fA$ and the observations $\{Y_r\}_{r=s+1}^t$. This set may be empty, if no such filter states exist. As $C$ is a diagonal matrix, if its entries do not vanish\footnote{This corresponds to there being no state which yields a zero likelihood of the observed value.} then it is invertible. However, the matrix $A$ will often not be an invertible matrix, so $(p)_t^{(\ofA,s)}$ is not generally a singleton.

\subsection{Filtering with Uncertain Priors}
We shall first consider the case where we assume the filtering equations apply, and the only uncertainty in our model is given by our uncertainty over the prior inputs to the filter. In particular, this ``prior uncertainty'' is not updated given new observations, and $\CR$ is a fixed function of $\bQ$

We can now consider our two cases: with static and dynamic generators.
\subsubsection{Static Generators (StaticUP)}
With a static generator, our approach is simple.  
\begin{definition}
In a StaticUP setting, the inputs to the filtering problem are the initial filter state $p_0$ and the generator $\fA$. We therefore take a penalty $\CR(\bQ) = (k^{-1} \alpha(\bQ))^{k'}$ where 
\[\alpha(\bQ) = \gamma(p_0^\bQ, \fA^\bQ)\]
for some prescribed penalty $\gamma$. 
\end{definition}

\begin{remark}
Inspired by the DR-expectation, a natural choice of penalty function $\gamma$ is the negative log-density of a prior distribution for the inputs $(p_0, \fA)$, shifted to have minimal value zero. Alternatively, taking an empirical Bayesian perspective, $\gamma$ could be the log-likelihood from a prior calibration process. In this case, we are directly using our statistical uncertainty in the parameters in our understanding of the filtering problem.
\end{remark}

\begin{lemma}\label{lem:staticUPK}
If the extended penalty is defined by
\[K_t(p_t, \fA) = \inf_{p_0\in (p_t)^{(\ofA,0)}}\Big\{\gamma(p_0, \fA)\Big\},\]
with the convention $\inf\emptyset=\infty$, then writing $\kappa_t(p) = \inf_{\mathfrak{A}\in \bA} K_t(p,\mathfrak{A})$, we satisfy the dynamic updating equation \eqref{eq:recursivePen}.
\end{lemma}
\begin{proof}
From \eqref{eq:recursivePen}, we see that it is enough to guarantee that 
\begin{align*}
\kappa_t(q) &= \inf_{\bQ\in \CM_M}\Big\{\alpha(\bQ): X_t|\CY_t \sim q \text{ under }\bQ\Big\}\\
&=\inf_{\fA\in\bA, p_0\in S_N^+} \Big\{\alpha(\bQ^{\fA, p_0}):  X_t|\CY_t \sim q \text{ under }\bQ\Big\}\\
&=\inf_{\fA\in\bA} \inf_{p_0\in (q)_t^{(\ofA,0)}} \Big\{ \gamma(p_0, \fA)\Big\} = \inf_{\fA\in\bA} K(q, \fA).
\end{align*}
The result then follows from \eqref{eq:kappaKDefn}.
%From Definition \ref{defn:backprojection},  $(p_t)^{(\ofA,0)}$ describes the time-reversed propagation of our filter state. Therefore, $K_t(p_t, \fA)$ describes the minimal penalty associated with a terminal filter-state $p_t$, for each set of parameters $\fA$. Therefore, $\kappa_t(p_t)$ corresponds to the minimal penalty over all choices of the filter inputs $(p_0, \fA)$ which lead to the filtered state $p_t$, as desired.
\end{proof}

The extended penalty $K$ is useful, as it can be calculated recursively.
\begin{theorem}\label{thm:fixUP}
The extended penalty $K$ satisfies the recursion
\[K_{t}(p,\fA) = \inf_{p_{t-1}\in (p)_t^{\ofA}}\Big\{K_{t-1}(p_{t-1},\fA)\Big\}.\]
\end{theorem}
\begin{proof}
This is immediate from the definition of $(p)_t^{\ofA}$.
\end{proof}

\begin{remark}
We emphasize that there is no learning of $\fA$ being done in this framework -- the penalty on $\fA$ applied at time $0$ is simply propagated forward; our observations do not affect our opinion of its likely value. Furthermore, we are not adjusting our prior penalty to account for the `unreasonableness' of our models as we observe data. In particular, if we assume no knowledge of the initial state (i.e. a zero penalty), then we will have no knowledge of the state at time $t$ (unless the observations cause the filter to degenerate).
\end{remark}

\begin{example}\label{example1}
We take the class of models in $\CM_M$ where $A$ and $C$ are perfectly known, and $A= I$, so $X_t=X_0$ is constant (but $X_0$ is unknown). We take $N=2$, so $X$ takes only one of two values. Finally, we assume that 
\[Y_t|(X_t=e_1) \sim \mathrm{Bernoulli}(a), \qquad Y_t|(X_t=e_2) \sim \mathrm{Bernoulli}(b),\]
where $a,b\in (0,1)$. Effectively, in this example we are using filtering to determine which of two possible means is the correct mean for our observation sequence. It is worth emphasising that the filter process $p$ corresponds to the posterior probabilities, in a Bayesian setting, of the events that our Bernoulli process has parameter $a$ or $b$.

It will be useful to note that, from classical Bayesian statistical calculations\footnote{One can derive the stated formula using the filtering equations, for the vector $p_t=(p_t^1, p_t^2)^\top$. However, the closed-form solution given here is more easily obtained using alternative methods for Bayesian hypothesis testing (which is effectively what this problem encodes).}, for a given $p_0$, one can see that the corresponding value of $p_t$ is determined from the log odds ratio
\[\log\Big(\frac{p_t^1}{p_t^2}\Big) = \log\Big(\frac{p_0^1}{p_0^2}\Big)+ {N\bar Y}\log\Big(\frac{a}{b}\Big)+N(1-\bar Y)\log\Big(\frac{1-a}{1-b}\Big) .\]

To write down the \emph{StaticUP} penalty function, let the (known) dynamics be described by $\fA^*$. Consequently, we can write  $K(p, \fA)=\infty$ for all $\fA\neq \fA^*$. As $\fA^*$ is known, there is no distinction between $K$ and $\kappa$, so 
\[\kappa_t(p) =K_t(p, \fA^*) = \inf_{p_0\in (p)_t^{(\overleftarrow{\fA^*},0)}}\Big\{\gamma(p_0, \fA^*)\Big\}.\] We initialize with a known penalty $\gamma(p, \fA^*)=\kappa_0(p)$ for all $p\in S_N^+$. %As $\fA^*=(A,C)$ with $A=I$ and $C$ nonvanishing, we have that $A$ and $C$ are invertible, so $(p_t)^{\overleftarrow{\fA^*}}$ will be a singleton or empty for every $p_t$.

In this setting, we can express our penalty in terms of the log-odds, for the sake of notational simplicity given the closed-form solution to the filtering problem, and hence can explicitly calculate $(p_t)^{(\overleftarrow{\fA^*},0)}$, which contains only single points. In particular, the time-$t$ penalty is given by a shift of the initial penalty:
\[\kappa_t\bigg(\log\Big(\frac{p_t^1}{p_t^2}\Big)\bigg) = \kappa_0\bigg(\log\Big(\frac{p_t^1}{p_t^2}\Big)-{N\bar Y}\log\Big(\frac{a}{b}\Big)-N(1-\bar Y)\log\Big(\frac{1-a}{1-b}\Big) \bigg).\]
\end{example}
\begin{remark}
This example demonstrates the following behaviour:
\begin{itemize}
\item If the initial penalty is zero, then the penalty at time $t$ is also zero -- there is no learning of which state we are in. 
\item There is no variation in the curvature of the penalty (and so no change in our `uncertainty'), we simply shift the penalty around, corresponding to our changing posterior probabilities.
\item The update of $\kappa$ is done purely using the tools of Bayesian statistics, rather than having any direct incorporation of our uncertainty.
\end{itemize}
\end{remark}
\begin{remark}
We point out that this is, effectively, the model of uncertainty proposed by Walley \cite{Walley1991} (see in particular Section 5.3, although he there takes a model where the unknown parameter is Beta distributed). See also Fagin and Halpern \cite{Fagin90}.
\end{remark}

\subsubsection{Dynamic generators (DynamicUP)}
If we model the generator $\fA$ as fixed and unknown, calculation of $K_t(p, \fA)$ suffers from a curse of dimensionality. We also need to assume that $\fA$ is constant through time, which may be a dubious assumption in practice. 

In this case, we can obtain a practical model by allowing $\fA$ to vary independently at each point in time. Superficially, this significantly worsens the curse of dimensionality, as we no longer take a fixed $\fA$, but regard it as a process through time. The advantage of this is that we can then use dynamic programming to calculate the penalty $\kappa_t(p_t)$. 

In order to include our knowledge of the generator, we will write $\fA_t$ for the generator applicable at time $t$. Recall that this is a process taking values in $\bA$, and we write $\CA_{\CY}$ for the space of such processes adapted to the observation filtration. For $\fA\in \CA_\CY$, we define $(p)_t^{(\ofA,0)}$ using the natural analogue of Definition \ref{defn:backprojection} incorporating time dependence.

\begin{definition}
In the DynamicUP setting, for an initial penalty on the initial hidden state, $\kappa_0(p_0)$, and a penalty on the time-$t$ generator, $\gamma_t(\fA_t)$, our total penalty is given by 
\[\alpha(\bQ) = \kappa_0(p_0^\bQ) + \sum_{s=1}^{\infty} \gamma_s(\fA_s^\bQ).\]
\end{definition}

\begin{theorem}\label{thm:dynamicUP}
We can define a dynamic penalty
\begin{equation}\label{eq:dynamRecurs0}
 \kappa_t(p_t) = \inf_{\fA\in \CA_\CY}\bigg\{\inf_{p_0 \in (p_t)^{(\{\ofA_t\},0)}}\Big\{\kappa_0(p_0) + \sum_{s=1}^{t} \gamma_s(\fA_s)\Big\}\bigg\}.
\end{equation}
which satisfies \eqref{eq:recursivePen}. Furthermore, the penalty satisfies the recursion
\begin{equation}\label{eq:dynamRecurs1}
 \kappa_t(p_t) = \inf_{\fA_t\in\bA}\bigg\{\inf_{p_{t-1}\in (p_t)^{\ofA_t}}\Big\{\kappa_{t-1}(p_{t-1}) + \gamma_t(\fA_t)  \Big\}\bigg\}.
\end{equation}
\end{theorem}
\begin{proof}
From Lemma \ref{lem:horizonirrelevant} it is clear that, with our definition of $\kappa_t$, \eqref{eq:recursivePen} can be obtained as a reparameterization of the nonlinear expectation. The recursion \eqref{eq:dynamRecurs1} follows by the definition of $(p)_t^{\ofA_t}$ and standard dynamic programming arguments, as in the StaticUP case.
\end{proof}

This formulation of our problem allows us to use dynamic programming to solve our problem forward in time. In the setting of Example \ref{example1}, as the dynamics $\fA$ are perfectly known, there is no distinction between the dynamic and static cases.

\begin{remark} The dynamic formulation adds penalties together, so if the same penalty on $\fA$ is used in the static and dynamic settings, then the dynamic setting will typically have a higher penalty than the static setting. Practically, this effect is lessened by the minimization in \eqref{eq:dynamRecurs0}, and the fact that the filter has good ergodic properties (as discussed by van Handel \cite{Handel2008}). This ergodicity implies that, at time $t$, the filter will not significantly depend on the generator $\fA_s$ for $s\ll t$, and so the minimization will render the penalty at $s$ irrelevant.
\end{remark}

 A continuous-time version of this setting is considered in \cite{Allan2017}.

\subsection{Filtering with DR-expectations}

In the above, we have regarded the prior as uncertain, and used this to penalize over models. We did not use the data to modify our penalty function, simply to revise \emph{within} each model. The DR-expectation appoach gives us an alternative approach, in which the data guides our model choice more directly. In what follows, we will apply the DR-expectation in our filtering context, and observe that it gives a slightly different recursion for the penalty function. Again, we can consider models where $\fA$ is regarded as static or as dynamic.

\subsubsection{Static generators (StaticDR)}

For a fixed $\fA$, we have already written the likelihood function (Proposition \ref{prop:LikePart}). As we are working in the observation filtration, this gives the natural decomposition of the likelihood for our calculation. This leads us to a simple formulation of the penalty $\alpha_{\mathcal{Q}|\mathbf{y}_t}$.

\begin{lemma}\label{lem:divergenceDRfilter}
Suppose the initial filter state $p_0$ and static generator $\fA$ are distributed according to the prior density $\exp(-\gamma(p_0, \fA))$. Then the $\mathcal{Q}|\mathbf{y}_t$-divergence of the measure $\bQ^{(p_0, \fA)}$ (i.e. normalized log-posterior density of the parameters $(p_0, \fA)$) given observations $\mathbf{y}_t = (Y_1,... Y_t)$  can be written
\[\alpha_{\mathcal{Q}|\mathbf{y}_t}(\bQ^{(p_0, \fA)}) = \gamma(p_0, \fA)-\sum_{s\leq t}\log\Big(c_\fA (Y_s; A^\fA p_{s-1}^{\fA, p_0})\Big) + m_t,\]
where $m$ is a sequence of normalizing values, independent of $\fA$ and $p_0$, such that $\alpha_{\mathcal{Q}|\mathbf{y}_t}$ has minimal value zero, and $p_s^{\fA, p_0}$ is the solution to the filtering equations with the given generator $\fA$ and initial filter state $p_0$.
\end{lemma}
\begin{proof}
 The divergence is given by the negative posterior log-density of $(p_0, \fA)$, shifted to have minimal value zero. From Bayes' rule, this is simply the sum of the negative prior log-density and the negative log-likelihood. The likelihood is stated in Proposition \ref{prop:LikePart}, and gives us the term $\sum_{s\leq t}\log\big(c_\fA (Y_s; A^\fA_s p_{s-1}^{\fA, p_0})\big)$. The prior then contributes the term $\gamma(p_0, \fA)$, as desired.
\end{proof}

\begin{remark}
 In this framework, we have not assumed we can propagate our uncertainty using the filtering equations, as in an uncertain prior setting -- we are instead simply evaluating the DR-expectation conditional on our observations. Nevertheless, the filtering equations naturally appear through their presence in the likelihood function, and so will still form part of our dynamic penalty.
\end{remark}

\begin{theorem}\label{thm:fixDR}
Consider the family of models with a fixed, but unknown, generator $\fA$. The DR-expectation can then be calculated as in \eqref{eq:recursivePen}, writing $\kappa_t(p) = \inf_{\mathfrak{A}\in \CA} K_t(p,\mathfrak{A})$, where 
\[K_t(p, \fA) := \inf_{p_0\in (p)_t^{(\ofA,0)}}\Big\{\gamma(p_0, \fA)-\sum_{s\leq t}\log\Big(c_\fA (Y_s; A^\fA_s p_{s-1}^{\fA, p_0})\Big) + m_t\Big\}.\]
Furthermore, $K$ has recursive representation
\[K_t(p,\fA) = \inf_{p_{t-1}\in (p)_t^{\ofA}} \Big\{K_{t-1}(p_{t-1}, \fA) -\log\Big(c_\fA (Y_t; A^\fA p_{t-1}\Big)\Big\} + m_t,\]
where $K(p_0, \fA) =\gamma(p_0, \fA)$  and $m_t$ is a constant to ensure $\inf_{p_t, \fA} K_t(p_t, \fA) =0$.
\end{theorem}
\begin{proof}
% From Proposition \ref{prop:LikePart}, we can write down the log-likelihood
% \[\ell(\{\CY_s\}_{s\leq t}|p_0,\fA) = \sum_{s\leq t}\log\Big(c_\fA (Y_s; A^\fA_s p_{s-1}^\fA)\Big) + k_t'\]
% where $k'$ is an adapted sequence of normalizing values, independent of $\fA$ and $p_0$, and $p_s^\fA$ is the solution to the filtering equations with the given generator $\fA$. To obtain the posterior log-density, we add the prior log-density $\ell(p_0, \fA)$, then multiply by $-1$ to obtain the penalty (up to a scale constant, which simply affects the value of $k_t'$).
As we are looking at a function of the current hidden state, we are interested in the minimal penalty associated with each $p_t$ (as any larger penalty will be ignored by the supremum when calculating the DR-expectation). Given our definition of $K$,  as in the StaticUP case (Lemma \ref{lem:staticUPK}), we observe that $\kappa_t(p) = \inf_{\fA\in\bA}\{K_t(p,\fA)\}$ satisfies \eqref{eq:recursivePen} by rearrangement. Using the definition of $(p)_t^{\ofA}$ and the recursive nature of the filter, by the usual dynamic programming arguments it is easy to see that we can calculate the penalty recursively, 
\[K_t(p,\fA) = \inf_{p_{t-1}\in (p)_t^{\ofA}} \Big\{K_{t-1}(p_{t-1}, \fA) -\log\Big(c_\fA (Y_t; A^\fA p_{t-1}\Big)\Big\} + m_t\]
with the initial value given by the prior penalty $K(p_0, \fA) = \gamma(p_0, \fA)$.
\end{proof}

\begin{remark}
Comparing the results of Theorem \ref{thm:fixUP} with Theorem \ref{thm:fixDR}, we see that the key distinction is the presence of the log-likelihood term $\log\big(c_\fA (Y_s; A^\fA_s p_{s-1}^\fA)\big)$. This term implies that observations of $Y$ will affect our quantification of uncertainty, rather than purely updating each model. We also note that the filtering equations have arisen naturally as a part of the penalty (through their appearance in the likelihood function adapted to $\CY$), rather than us assuming that the filtering equations are the `correct' method for updating in the presence of uncertainty.
\end{remark}

\begin{example}
In the setting of Example \ref{example1}, recall that $X$ is constant, so we know $\fA$. One can calculate the penalty either directly, or through solving the stated recursion using the dynamics of $p$. The resulting penalty is given by first calculating $p_0$ from $p_t$ through
\[\log\Big(\frac{p_0^1}{p_0^2}\Big) = \log\Big(\frac{p_t^1}{p_t^2}\Big)-{N\bar Y}\log\Big(\frac{a}{b}\Big)-N(1-\bar Y)\log\Big(\frac{1-a}{1-b}\Big)\]
and then 
\[\kappa_t(p_t) = \kappa_0(p_0)- \log\Big(p_0^1 a^{N\bar Y}(1-a)^{N(1-\bar Y)} + p_0^2 b^{N\bar Y}(1-b)^{N(1-\bar Y)}\Big)-m_t, \]
where $m_t$ is chosen to ensure $\inf_p \kappa_t(p)=0$. From this, we can see that the likelihood will modify our uncertainty directly, rather than us simply propagating each model via Bayes' rule. A consequence of this is that, if we start with extreme uncertainty ($\kappa_0\equiv 0$), then our observations will teach us what models are reasonable, thereby reducing our uncertainty (i.e. we will find $\kappa_t(p)>0$ for $p\in(0,1)$ when $t>0$).
\end{example}

\subsubsection{Dynamic generators (DynamicDR)}

As in the uncertain priors case, it is often practically inconvenient to calculate a recursion for $K(p, \fA)$, given the high dimension of $\fA$. We can avoid this by allowing $\fA$ to vary through time. In order to place this within the DR expectation framework, we consider the following models.

\begin{definition}\label{def:dynamicgenInd}
% Consider a model where $\{\fA_t\}$ is an iid sequence from $\CA$, and $\fA_t$ is selected according to the distribution with density $\exp(-\gamma(\fA))$ . 

Suppose $\{\fA_t\}_{t>0}$ is a sequence, selected according to distributions with conditional density 
\[\fA_t|\{\fA_s, Y_s\}_{s<t}\sim\exp(-\gamma(\fA_t; \{Y_s\}_{s<t})),\] for $\gamma$ a known function. Note that this structure implies $\{\fA_t\}_{t>0}\in \CA_\CY$.

In addition, suppose that $p_0$ is selected from a distribution with density $\exp(-\pi(p_0))$ on $S_N^+$, independently of $\{\fA_t\}_{t>0}$.
\end{definition}
This assumption allows us to decouple the choice of generators at different times, leading to a significant reduction in complexity. The prior penalty, represented by $(\gamma, \pi)$ is assumed to be known, potentially from calibration of the model.

\begin{theorem}\label{thm:dynamicDR}
Consider models where $p_0$ and $\{\fA_t\}_{t>0}$ satisfy Definition \ref{def:dynamicgenInd}. In a DR-expectation setting, a dynamic penalty $\kappa$ satisfying \eqref{eq:recursivePen} can be obtained from the recursion
\[\begin{split}\kappa_t(p_t) =\inf_{\fA_t}\bigg\{\inf_{p_{t-1}\in (p_t)^{\ofA_t}}\Big\{\kappa_{t-1}(p_{t-1}) &+ \gamma_t(\fA_t; \{Y_s\}_{s<t}) \\&-\log\Big(c_\fA (Y_t; A^\fA_t p_{t-1})\Big)\Big\}\bigg\}-m_t,\end{split}\]
with initial value $\kappa_0(p)=\pi(p)$, where $m_t$ is chosen to ensure $\inf_{p\in S_N^+}\kappa_t(p)=0$ for all $t$.
\end{theorem}
\begin{proof}
Using the same logic as in Lemma \ref{lem:divergenceDRfilter}, and using Lemma \ref{lem:horizonirrelevant} to restrict our horizon to $t$, we observe that the divergence, for a model $\bQ$ with generator $\fA = \{A_s,C_s\}_{s>0}$ is given by 
\[\alpha_{\mathcal{Q}|\mathbf{y}_t}(\bQ) =  \pi(p_0) +\sum_{0<s\leq t}\bigg(-\log\Big(c_\fA (Y_s; A_s p_{s-1}^{\fA, p_0})\Big) + \gamma_s(\fA_s;\{Y_r\}_{r<s})\bigg) -m_t,\]
where $p_s^{\fA, p_0}$ is the corresponding solution to the filtering equations, and $m_t$ is a normalizing constant.

As in the static generator case, we can then reparameterize in terms of the terminal value of the filter, and notice that we are only interested in the minimal penalty for a given $p_t$. Comparing with \eqref{eq:recursivePen}, we have
\[
\begin{split}\kappa_t(p) = \inf_{\fA\in \bA}\bigg\{&\inf_{p_0\in (p)_t^{(\ofA, 0)}} \bigg\{\pi(p_0) \\
&+\sum_{0<s\leq t}\bigg(-\log\Big(c_\fA (Y_s; A^\fA_s p_{s-1}^{\fA, p_0})\Big) + \gamma_s(\fA_s;\{Y_r\}_{r<s})\bigg)\bigg\}\bigg\}-m_t.\end{split}
\]
The recursion then follows by the usual dynamic programming arguments.
\end{proof}

\begin{remark}
We expect that there will be less difference between the dynamic uncertain prior and dynamic DR-expectation settings than between the static uncertain prior and static DR-expectation settings. This is because there is only limited learning possible in the dynamic DR-expectation, as the underlying generator is independently given at every time, so the DR-expectation has only one value with which to infer its behaviour. This increases the relative importance of the prior term $\gamma$, which describes our understanding of typical values of the generator. In practice, the key distinction between the dynamic DR-expectation and uncertain prior models appears to be when the initial penalty is near zero -- in this case, the DR-expectation regularizes the initial state quickly, while the uncertain prior model may remain near zero indefinitely.
\end{remark}

\begin{example}
In the setting of Example \ref{example1}, as the dynamics are perfectly known, there is again no difference between the dynamic and static generator DR-expectation cases.
\end{example}

\begin{remark}
In these sections, we have considered the case where the generator is either dynamic or static. Of course, further variations can be had, by allowing the generator to have some parts static and others dynamic, or to be  dynamic but with penalty determined by an unknown static parameter. These perturbations give this approach a high degree of flexibility for practical modelling.
\end{remark}

\section{Expectations of the future}\label{sec:consist}

 The nonlinear expectations considered above do not consider how the future will evolve. In particular, we have focussed our attention on calculating $\CE_{\mathcal{Q}|\mathbf{y}_t}(\phi(X_t))$, that is, on calculating the expectation of functions of the \emph{current} hidden state. In other words, we can consider our nonlinear expectation as a mapping
 \[\CE_{\mathcal{Q}|\mathbf{y}_t}: L^\infty(\sigma(X_t)\otimes \CY_t) \to L^\infty(\CY_t).\]
 
 If we wish to calculate expectations of \emph{future} states, then we may wish to consider doing so in a filtration-consistent manner. This is of particular importance when considering optimal control problems. 

 \begin{definition}\label{defn:dynamicnonlinearexpectation}
  For a fixed horizon $T>0$, suppose that for each $t< T$ we have a mapping $\CE(\cdot|\CY_t):L^\infty(\CY_T) \to L^\infty(\CY_t)$. We say that $\CE$ is a $\CY$-consistent convex expectation if $\CE(\cdot|\CY_t)$ satisifes the following assumptions, analogous to those above,
 \begin{itemize}
  \item Strict Monotonicity: for any $\xi_1, \xi_2\in L^\infty(\CY_T)$, if $\xi_1\geq \xi_2$ a.s. then $\CE(\xi_1|\CY_t) \geq \CE(\xi_2|\CY_t)$ a.s. and if in addition $\CE(\xi_1|\CY_t)=\CE(\xi_2|\CY_t)$ then $\xi_1=\xi_2$ a.s.
  \item Constant triviality: for  $K\in L^\infty(\CY_t)$, $\CE(K|\CY_t)=K$.
  \item Translation equivariance: for any $K\in L^\infty(\CY_t)$, $\xi\in L^\infty(\CY_T)$, $\CE(\xi+K|\CY_t)= \CE(\xi|\CY_t)+K$.
  \item Convexity: for any $\lambda\in [0,1]$, $\xi_1, \xi_2\in L^\infty(\CY_T)$, 
 \[\CE(\lambda \xi_1+ (1-\lambda) \xi_2|\CY_t) \leq \lambda \CE(\xi_1|\CY_t)+ (1-\lambda) \CE(\xi_2|\CY_t)\]
  \item Lower semicontinuity: For a sequence $\{\xi_n \}_{n\in\bN}\subset L^\infty(\CY_T)$ with $\xi_n \uparrow \xi \in L^\infty(\CY_T)$ pointwise, $\CE(\xi_n|\CY_t) \uparrow \CE(\xi|\CY_t)$ pointwise for every $t<T$.
 \end{itemize}
 and the additional asssumptions
 \begin{itemize}
  \item $\{\CY_t\}_{t\ge 0}$-consistency: for any $s<t<T$, any $\xi\in L^\infty(\CY_T)$,
 \[\CE(\xi|\CY_s) = \CE(\CE(\xi|\CY_t)|\CY_s).\]
 \item Relevance: for any $t<T$, any $A\in \CY_t$, $\CE(I_A\xi|\CY_t) = I_A \CE(\xi|\CY_t)$.
 \end{itemize}
 \end{definition}
 The assumption of $\CY$-consistency is sometimes simply called recursivity, time consistency or dynamic consistency (and is closely related to the validity of the dynamic programming principle), however, it is important to note that this depends on the choice of filtration. In our context, consistency with the \emph{observation} filtration $\CY$ is natural, as this describes the information available for us to make decisions.

\begin{remark}Definition \ref{defn:dynamicnonlinearexpectation} is equivalent to considering a lower semicontinuous convex expectation, as in Definition \ref{defn:nonlinearexpectation} and assuming that for any $\xi\in L^\infty(\CY_T)$ and any $t<T$, there exists a random variable $\xi_t$ such that $\CE(I_A \xi) = \CE(I_A \xi_t)$ for all $A\in \CY_t$. In this case, one can define $\CE(\xi|\CY_t) = \xi_t$ and verify that the definition given is satisfied (see F\"ollmer and Schied \cite{Follmer2002}).
\end{remark}

Much work has been done on the construction of dynamic nonlinear expectations (see for example Epstein and Schneider \cite{Epstein2003}, Duffie and Epstein \cite{Duffie1992}, El Karoui, Peng and Quenez \cite {El1997}, Cohen and Elliott \cite{Cohen2008c}, and references therein). In particular, there have been close relations drawn between these operators and the theory of BSDEs (for a setting covering the discrete-time examples we consider here, see \cite{Cohen2008c,Cohen2009a}).

\begin{remark}The importance of $\CY$-consistency is twofold: First, it guarantees that, when using a nonlinear expectation to construct the value function for a control problem, an optimal policy will be consistent in the sense that (assuming an optimal policy exists) a policy which is optimal at time zero will remain optimal in the future. Secondly, $\{\CY_t\}_{t\ge 0}$-consistency allows the nonlinear expectation to be calculated recursively, working backwards from a terminal time. This leads to a considerable simplification numerically, as it avoids a curse of dimensionality in intertemporal control problems. \end{remark}

\begin{remark}
One issue in our setting is that our lack of knowledge does not simply line up with the arrow of time -- we are unaware of events which occurred in the past, as well as those which are in the future. This leads to delicacies in questions of dynamic consistency. Conventionally, this has been often been considered in a setting of  `partially observed control', and these issues are resolved by taking the filter state $p_t$ to play the role of a state variable, and solving the corresponding
`fully observed control problem' with $p_t$ as underlying. In our context, we do not know the value of $p_t$, instead we have the (even higher dimensional) penalty function $\kappa_t$ (or worse, $K_t$) as a state variable (taking values in the space of functions on $S_N^+$).
\end{remark}

In this section, we will outline how this perspective can provide a dynamically consistent extension of our expectations,  and how enforcing dynamic consistency will modify our perception of risk. For this and the remainder of the paper, we will focus our attention on the dynamic-generator DR-expectation framework. The corresponding theory in the dynamic-generator uncertain-prior setting can be obtained by simply removing the relevant likelihood term whenever it appears.

\subsection{Asynchronous expectations}

We will focus our attention on constructing a dynamically consistent nonlinear expectation for random variables in $L^\infty(\sigma(X_T)\otimes \CY_T)$, given observations up to times $t<T$. Given our earlier analysis, we already have a map 
\[\CE_{\mathcal{Q}|\mathbf{y}_T}: L^\infty(\sigma(X_T)\otimes \CY_T) \to L^\infty(\CY_T).\] We therefore need only to construct  a $\CY$-consistent family of maps 
\[\oCE(\cdot|\CY_t): L^\infty(\CY_T) \to L^\infty(\CY_t),\] which we can extend by composition with $\CE_{\mathcal{Q}|\mathbf{y}_T}$ to be defined on our space of interest $L^\infty(\sigma(X_T)\otimes \CY_T)$.

As we are in discrete time, we can construct a $\CY$-consistent family through recursion. Therefore, the key question is how we construct a nonlinear expectation over one step. The definition of the DR-expectation can be applied to generate these one-step expectations in a natural way.

Recall that, as $\CY$ is generated by $Y$, any $\CY_{t+1}$-measurable function $\xi$ is simply a function of $\{Y_s\}_{s\leq t+1}$ so we can write
\begin{equation}\label{eq:xiMarkovDef}
 \xi(\omega) = \hat\xi(Y_{t+1}, \{Y_s\}_{s\leq t}).
\end{equation}
For any conditionally Markov measure $Q$,  if $Q$ has generator $\{\fA_t\}_{t> 0}$, it follows that
\[\bE_\bQ[\xi|\CY_t]= \int_{\bR^d} \hat\xi(y, \{Y_s\}_{s\leq t})   c_{\fA_{t+1}}(dy; A^\fA_t p_t ).\]

\begin{lemma}
For a dynamic DR-expectation, the one-step expectation (i.e. for $\CY_{t+1}$-measurable $\xi_{t+1}$) can be written
\begin{align*} 
\CE(\xi_{t+1}|\CY_{t})&:= \CE_{\mathcal{Q}|\mathbf{y}_t}(\xi_{t+1})\\
& = \sup_{p_t\in S_N^+, \fA_t\in \CA}\Big\{ \int_{\bR^d} \hat\xi_{t+1}(y, \{Y_s\}_{s\leq t})   c_{\fA_{t+1}}(dy; A^\fA_t p_t )\\&\qquad\qquad-\Big(\frac{1}{k}(\kappa_t(p_t)+\gamma_t(\fA_t; \{Y_s\}_{s<t})\Big)^{k'} \Big\},
\end{align*}
where $\hat\xi$ is as in \eqref{eq:xiMarkovDef}.
\end{lemma}
\begin{proof}
Our DR-expectation can be written 
\[ \CE_{\mathcal{Q}|\mathbf{y}_t}(\xi_{t+1}) = \sup_{\bQ}\Big\{E_\bQ[\xi_{t+1}|\mathbf{y}]-\Big(\frac{1}{k}\alpha_{\mathcal{Q}|\mathbf{y}_t}(\bQ)\Big)^{k'} \Big\}.\]
As $\xi_{t+1}$ is $\CY_{t+1}$-measurable, rather than $\CY_t$-measurable, we cannot exploit Lemma \ref{lem:horizonirrelevant} fully, as the generator at $t+1$ is still relevant when calculating our $\CY_t$-conditional expectation. As we are in a dynamic-generator DR-expectation setting, for a measure $\bQ$ with generator $\{\fA_s\}_{0< s\le T}$ and initial hidden distribution $p_0$, we have
 \[\alpha_{\mathcal{Q}|\mathbf{y}_t}(\bQ) =  \pi(p_0) -\sum_{0<s\leq t}\log\Big(c_\fA (Y_s; A_s p_{s-1}^{\fA, p_0})\Big) +\sum_{0< s\leq T}\gamma_s(\fA_s;\{Y_r\}_{r<s}) -m_t,\]
Our expectation $\bE_\bQ[\xi_{t+1}|\CY_t]$ is independent of $\fA_s$ for $s>t+1$, and is independent of $\fA_s$ for $s\le t$ given $p_t$. Therefore, without loss of generality we can minimize this over possible values of $\{\fA_s\}_{s\neq t+1}$ and, expressing in terms of the current filter state, we obtain 
\[\inf_{\{\fA_s\}_{s\neq t+1}}\Big\{\inf_{p_0\in (p)_t^{(\{\ofA_s\}_{s<t}, 0)}}\big\{\alpha_{\mathcal{Q}|\mathbf{y}_t}(\bQ)\big\}\Big\} = \kappa_t(p_t) + \gamma_{t+1}(\fA_{t+1}),\]
where $\kappa_t$ is as constructed in Theorem \ref{thm:dynamicDR}.
Combining the conditional expectation and the penalty, this gives the desired representation.
\end{proof}

\begin{remark}
 There is a surprising form of double-counting of the penalty here. For notational simplicity, let's assume $\phi$ does not depend on $Y$. If we consider $\xi_{t+1} = \CE(\phi(X_{t+1})|\CY_{t+1})$, then we have included a penalty for the proposed model at $t+1$, that is,
\[\xi_{t+1} = \CE(\phi(X_{t+1})|\CY_{t+1}) = \sup_{p\in S_N^+}\Big\{\sum_i p^i\phi(e_i)-\Big(\frac{1}{k}(\kappa_{t+1}(p)\Big)^{k'}\Big\}\]
where $\kappa_{t+1}(p)$ is the penalty associated with the filter state at time $t+1$, which comes from the generators $(p_0, \{\fA_s\}_{s\le t+1})$.

  When we calculate $\CE(\xi_{t+1}|\CY_{t})$, we then add on the penalty $\kappa_t(p_t)+ \gamma_{t+1}(\fA_{t+1})$, which again penalizes unreasonable values of $p_t$ and the generator $\fA_{t+1}$. This `double counting' of the penalty corresponds to us including both our `expected' uncertainty at time $t+1$, and also our `uncertainty at $t$ about our uncertainty at $t+1$'. 
\end{remark}

\begin{remark}
In the case where we considered a DynamicUP model, the equations would be identical, with the corresponding choice of $\gamma$. This is because the DR and uncertain prior models differ only through the incorporation of learning through the log-likelihood term, which is not a consideration when it comes to evaluating our future expectations.
\end{remark} 

\begin{remark}
If we take a StaticDR model, then the equations vary as one might expect:
\[ \CE(\xi_{t+1}|\CY_{t}) = \sup_{p_t\in S_N^+, \fA\in\CA}\Big\{\int_{\bR^d} \hat\xi(y, \{Y_s\}_{s\leq t})   c_{\fA_{t+1}}(dy; A^\fA_t p_t )-\Big(\frac{1}{k}K_t(p_t, \fA)\Big)^{k'}\Big\}.\]
Similarly for a StaticUP model. However, one should be careful in this setting, as the recursively-defined nonlinear expectation will consider models which allow the generator $\fA$ to vary through time, even though this does not form part of the original static generator framework.
\end{remark}

Given this recursion, we can now define the nonlinear expectation at every time.

\begin{definition}
 The dynamically consistent expectation of $\phi(X_T, \{Y_t\}_{t\le T})$ (for either the static or dynamic generator cases, and either the uncertain prior or DR-expectation models), is given by the recursion
\[\oCE(\phi(X_T, \{Y_t\}_{t\le T})|\CY_t):=\xi_t = \CE(\xi_{t+1}|\CY_t)\]
with $\xi_T = \CE_{\mathcal{Q}|\mathbf{y}_T}(\phi(X_T, \{Y_t\}_{t\le T}))$. As $\CY_0$ is trivial, we identify 
\[\oCE(\phi(X_T, \{Y_t\}_{t\le T})) = \oCE(\phi(X_T, \{Y_t\}_{t\le T})|\CY_0).\]
\end{definition}
 Note that $\xi_t$ is $\CY_t$-measurable for all $t<T$. 

\begin{remark}
 As we have defined $\xi$ using recursion, and our DR-expectation (and uncertain prior expectation) are convex, it is easy to verify that the map $\phi(X_T, \{Y_t\}_{t\le T})\mapsto \xi_t$ is a $\CY$-consistent convex expectation.
\end{remark}

\subsection{Recall of BSDE theory}

While it is useful to give a recursive definition of our nonlinear expectation, a better understanding of its dynamics is of practical importance. In what follows we will, for the dynamic generator case, consider the corresponding BSDE theory, assuming that $Y_t$ can take only finitely many values, as in \cite{Cohen2008c}. We will now present the key results of \cite{Cohen2008c}, in a simplified setting.

In what follows, we suppose that $Y$ takes $d$ values, which we associate with the standard basis vectors in $\bR^{d}$. For simplicity, we write $\bone$ for the vector in $\bR^{d}$ with all components $1$.

\begin{definition}
Write $\bar\bP$ for a probability measure such that $\{Y_t\}_{t\ge 0}$ is an iid sequence, uniformly distributed over the $d$ states, and $M$ for the $\bar\bP$-martingale difference process $Y_t-d^{-1}\bone$. As in \cite{Cohen2008c}, $M$ has the property that any $\CY$-adapted $\bar\bP$-martingale $L$ can be represented by $L_t= L_0 + \sum_{0\le s<t}Z_s M_{s+1}$ for some $Z$ (and $Z$ is unique up to addition of a multiple of $\bone$). 
\end{definition}

\begin{remark}
 The construction of $Z$ in fact also shows that, if $L$ is written $L_t= \tilde L(Y_1,...,Y_{t-1}, Y_t)$, then $e_i^\top Z_t = L(Y_1,..., Y_{t-1}, e_i)$ for every $i$ (up to addition of a multiple of $\bone$).
\end{remark}

We can then define a BSDE (Backward Stochastic Difference Equation) with solution $(\xi, Z)$:
\begin{equation}\label{eq:BSDEgen}
\xi_t(\omega) - \sum_{t\leq u< T} f(\omega, u, \xi_u(\omega), Z_u(\omega)) + \sum_{t\leq u< T} Z_u(\omega) M_{u+1}(\omega) = \xi_T(\omega),
\end{equation}
where $T$ is a finite deterministic terminal time, $f$ a $\CY$-adapted map $F:\Omega\times \{0,...,T\} \times \mathbb{R} \times \mathbb{R}^{d}\rightarrow \mathbb{R}$, and  $\xi_T$ a given $\mathbb{R}$-valued $\mathcal{Y}_T$-measurable terminal condition. For simplicity, we henceforth omit the $\omega$ argument of $\xi, Z$ and $M$.

The general existence and uniqueness result for BSDEs in this context is as follows:
\begin{theorem}\label{thm:BSDEExist}
Suppose $f$ is such that the following two assumptions hold:
\begin{enumerate}[(i)]
	\item For any $\xi$, if $Z^1= Z^2+k\bone$ for some $k$, then $f(\omega,t, \xi_t, Z^1_t) = f(\omega, t, \xi_t, Z^2_t)$, $\bar\bP$-a.s. for all $t$.
	\item For any $z\in\mathbb{R}^{d}$, for all $t$, for $\bar\bP$-almost all $\omega$, the map 
	\[\xi\mapsto \xi-f(\omega, t, \xi, z)\] is a bijection $\mathbb{R}\rightarrow \mathbb{R}$.
\end{enumerate}
Then for any terminal condition $\xi_T$ essentially bounded, $\mathcal{Y}_T$-measurable, and with values in $\mathbb{R}$, the BSDE (\ref{eq:BSDEgen}) has a $\CY$-adapted solution $(\xi, Z)$. Moreover, this solution is unique up to indistinguishability for $\xi$ and indistinguishability up to addition of multiples of $\bone$ for $Z$.
\end{theorem}

In this setting, we also have a comparison theorem:

\begin{theorem}\label{thm:CompThm}
Consider two discrete time BSDEs as in (\ref{eq:BSDEgen}), corresponding to coefficients $f^1, f^2$ and terminal values $\xi^1_T, \xi^2_T$. Suppose the conditions of Theorem \ref{thm:BSDEExist} are satisfied for both equations, let $(\xi^1, Z^1)$ and $(\xi^2, Z^2)$ be the associated solutions. Suppose the following conditions hold:
\begin{enumerate}[(i)]
	\item $\xi^1_T\geq \xi^2_T$ $\bar\bP$-a.s.
	\item $\bar\bP$-a.s.,  for all times $t$ and every $\xi\in \bR$ and $z\in \bR^{d}$, \[f^1(\omega, t, \xi, z) \geq f^2(\omega, t, \xi, z).\]
	\item $\bar\bP$-a.s., for all $t$, $f^1$ satisfies
	\[f^1(\omega, t, \xi_t^2, Z_t^1) - f^1(\omega, t, \xi_t^2, Z_t^2)\geq\min_{j\in \mathbb{J}_t}\{(Z^1_t-Z^2_t)(e_j-d^{-1}\bone)\}.\]
	\item $\bar\bP$-a.s., for all $t$ and all $z\in \bR^{d}$, $\xi\mapsto \xi-f^1(\omega, t, \xi,z)$ is an increasing function.
\end{enumerate}
 It is then true that $\xi^1 \geq \xi^2$ $\bar\bP$-a.s. A driver $f^1$ satisfying (iii) and (iv) will be called `balanced'.
\end{theorem}

Finally, we also know that all dynamically consistent nonlinear expectations can be represented through BSDEs:
\begin{theorem}
 The following two statements are equivalent.
 \begin{enumerate}[(i)]
 \item $\oCE(\cdot|\mathcal{Y}_t)$ is a $\mathcal{Y}_t$-consistent, dynamically translation invariant, nonlinear expectation 
 \item There exists a driver $f$ which is balanced, independent of $\xi$, and satisfies the normalisation condition $f(\omega, t, \xi_t, 0) = 0$, such that, for all $\xi_T$, $\xi_t = \oCE(\xi_T|\mathcal{Y}_t)$ is the solution to a BSDE with terminal condition $\xi_T$ and driver $f$.
 \end{enumerate}
 Furthermore, these two statements are related by the equation
 \[f(\omega, t, \xi, z) =\oCE(zM_{t+1}|\CY_t).\]
 \end{theorem}

\subsection{BSDEs for forward expectations}
By applying the above general theory, we can easily see that our nonlinear expectation has a representation as the solution to a particular BSDE. 
\begin{theorem}
In the dynamic generator setting, writing $\oCE(\phi(X_T, \{Y_t\}_{t\le T})|\CY_t) = \xi_t$, the dynamically consistent expectation satisfies the BSDE 
\[\xi_{t+1} = \xi_t-f(Z_t; \kappa_t) + Z_t M_{t+1}\]
where 
\[f(Z_t; \kappa_t) = \sup_{p, \fA}\Big\{\sum_i \Big(Z^i \big(c_\fA(e_i; A^\fA p)- d^{-1}\big)\Big) - \Big(\frac{\kappa_{t}(p) + \gamma_{t+1}(\fA)}{k}\Big)^{k'}\Big\}.\]
\end{theorem}
\begin{proof}
As $\xi_{t+1}$ is $\CY_{t+1}$-measurable, by the Doob--Dynkin lemma there exists a $\CY_t$-measurable function $\hat\xi_{t+1}$ such that $\xi_{t+1}=\hat\xi_{t+1}(Y_{t+1})$ (we omit to write $\{Y_s\}_{s\le t}$ as an argument).  We write $Z_t$ for the vector containing each of the values of this function. From the definition of $M$, as in the proof of the martingale representation theorem in \cite{Cohen2008c}, it follows that,
\[\xi_{t+1} - \bE_{\bar\bP}[\xi_{t+1}|\CY_t] = Z_t M_{t+1}.\]
 We can then calculate
\[\begin{split}
&  \xi_t  - \bE_{\bar\bP}[\xi_{t+1}|\CY_t]\\
&= \CE(\xi_{t+1}|\CY_t) - \bE_{\bar\bP}[\xi_{t+1}|\CY_t]\\
&= \sup_{p, \fA}\Big\{\sum_y \hat\xi_{t+1}(y) c_\fA(y; A^\fA p) - \Big(\frac{\kappa_{t}(p) + \gamma_{t+1}(\fA)}{k}\Big)^{k'}\Big\} - \bE_{\bar\bP}[\xi_{t+1}|\CY_t]\\
&= \sup_{p, \fA}\Big\{\sum_y \hat\xi_{t+1}(y) \big(c_\fA(y; A^\fA p)- d^{-1}\big) - \Big(\frac{\kappa_{t}(p) + \gamma_{t+1}(\fA)}{k}\Big)^{k'}\Big\}\\
&= f(Z_t; \kappa_t).
  \end{split}
\]
The answer follows by rearrangement.
\end{proof}

We can now observe that, if we treat the function $\kappa_t: S_N^+\to \bR$ as a state variable, we obtain a `Markovian' BSDE.

\begin{theorem}
 Suppose $\phi$ and $\gamma$ are independent of $\{Y_t\}_{t\le T}$. Then the solution to the above BSDE is a functional of $\kappa_t$, that is, $\xi_{t+1}(\omega) = \Xi_{t+1}(\kappa_t(\omega,\cdot))$. 
\end{theorem}
\begin{proof}
 This argument follows in the usual manner -- we observe that $\kappa$ has recursive dynamics (and furthermore, under the reference measure, $\kappa$ is Markovian) and that the terminal value of our BSDE is a function of $\omega$ only through $\kappa_T$. Consequently, we can use backward induction to construct the solution to the BSDE as a function of $\kappa_t$ (by evolving $\kappa$ forward one step in time, then solving the BSDE backward one step), and the result follows.
\end{proof}

\section{A control problem with uncertain filtering}

In this final section, we will consider the solution of a simple control problem under uncertainty, using the formal structures previously developed. We shall focus our attention on a DR-expectation with dynamic generator, however similar arguments can be used in each of the other settings considered. In some ways, this approach is similar to those considered by Bielecki, Chen and Cialenco \cite{Bielecki2016}, where the DR-expectation is replaced by an approximate confidence interval. (Taking $k'=\infty$ in our analysis would give a very similar problem to the one they consider.)

Suppose a controller selects a control $u$ from a set $U$, which we assume is a countable union of compact metrizable sets (this assumption is purely to enable us to use appropriate measurable selection theorems). Controls are required to be $\CY$-predictable (i.e. $u_t$ is $\CY_{t-1}$-measurable), and we write $\mathcal{U}$ for the space of such controls.

A control has an impact on the generator of $X,Y$, through modifying the penalty function $\gamma$, which describes the `reasonable' models for the transition matrix $A$ and the distribution of observations $c$. In particular, for a given $u$ we will now have a penalty $\gamma_t(\fA_t;u_t)$, which we assume is continuous in $u_t$ for every $\fA_t\in\bA$. This allows a controller to modify what are `reasonable' values of the generator, even though the generator may not be fully known. We write $\oCE_u$ for the corresponding $\CY$-consistent expectation and $\kappa^u$ for the corresponding dynamic penalty. The expectation $\oCE_u$ then encodes both the change in the hidden dynamics in the future due to future controls and the change in our agent's understanding of the present hidden state (as represented via $\kappa^u$) due to her past controls.

The controller wishes to minimize an expected  cost
\[\oCE_u\Big(\mathfrak{C}(X_T, \{Y_s\}_{s\leq T}) + \sum_{t< T} \mathfrak{L}_t(\{Y_s\}_{s< t}, u_{t+1})\Big).\]
Here $\mathfrak{C}$ is a terminal cost, which may depend on the hidden state $X_T$, and $\mathfrak{L}$ is a running cost, which will depend on the control $u_{t+1}$ used at time $t$. We assume $\mathfrak{C}$ and $\mathfrak{L}$ are continuous in $u$ (almost surely). We do not allow $\mathfrak{L}_t$ to depend on $X_t$, as this would potentially lead to paradoxes (as the agent could learn information about the hidden state by observing their running costs). We think of the cost $\mathfrak{L}_t$ as being paid at time $t$, depending on the choice of control $u_{t+1}$ (which will affect the generator at time $t+1$). For notational simplicity, we will omit to write $Y$ as an argument when unnecessary. 

For a given control process $u$, we define the remaining cost
\[J(\omega, t, u) = \oCE_u\Big(\mathfrak{C}(X_T) + \sum_{t\le s < T} \mathfrak{L}_s(u_{s+1})\Big|\CY_t\Big)\]
and hence the value function 
\[V(\omega, t) = \inf_{u\in \mathcal{U}}\oCE_u\Big(\mathfrak{C}(X_T) + \sum_{t\le s < T} \mathfrak{L}_s(u_{s+1})\Big|\CY_t\Big).\]

\begin{remark}\label{rem:strotz}
We define our expected cost using the $\CY$-consistent expectation $\oCE_u$, rather than the (inconsistent) DR-expectation $\CE_{\mathcal{Q}|\mathbf{y}_t}$), as this leads to time-consistency in the choice of controls.
\end{remark}

\begin{remark}
 We can see that the calculation of the value function is a `minimax' problem, in that $V$ minimizes the cost, which we evaluate using a maximum over a set of models. However, given the potential for learning, the requirement for time consistency, and the uncertainties involved, it is not clear that one can write $V$ explicitly in terms of a single minimization and maximization of a given function.
\end{remark}

\begin{remark}As the filter-state penalty $\kappa$ is a general function depending on the control, and $Y$ only takes finitely many states, it is not generally possible to express the effect (on $\kappa$) of a control through a change of measure relative to some reference dynamics. In particular, we face the problem that controls $u_s$ for times $s<T$ will have an impact on the terminal cost $V_T= \oCE_u(\mathfrak{C}(X_T)|\CY_T)$, through their impact on the uncertainty $\kappa^u_T$, so, unlike in a traditional control problem, $V_T$ is not independent of $u$ given $\CY_T$. For this reason, even though we model the impact of a control through its effect on the generator, we cannot give a fully `weak' formulation of our control problem, and are restricted to a `Markovian' setting, where we shall exploit dynamic programming. 
\end{remark}
\begin{theorem}
The value function satisfies a dynamic programming principle, in particular, if an optimal control $u^*$ exists, then for every $t\le T$,
\begin{align*}
 V_{t-1} &= \oCE_{u^*}(V_t|\CY_{t-1}) +\mathfrak{L}_{t-1}(u_{t}^*)
\end{align*}
(and similarly if we only assume an $\epsilon$-optimal control exists for every $\epsilon>0$).
\end{theorem}
\begin{proof}
For any control $u$, using the recursivity of $\oCE$ we have
\[\begin{split}
J(\omega,t-1,u) &= \oCE_u\Big(\mathfrak{C}(X_T) + \sum_{t-1\le s \le T} \mathfrak{L}_s(u_{s+1})\Big|\CY_{t-1}\Big)\\
&=\oCE_u\Big(\mathfrak{C}(X_T) + \sum_{t\le s \le T} \mathfrak{L}_s(u_{s+1})\Big|\CY_{t-1}\Big)+ \mathfrak{L}_{t-1}(u_{t})\\
&=\oCE_u\big(J(\omega, t, u)\big|\CY_{t-1}\big)+ \mathfrak{L}_{t-1}(u_t).
\end{split}\]
The stated equality (which is the dynamic programming relation) then follows from a standard pasting argument (along with a measurable selection result to ensure optimal or $\epsilon$-optimal controls exist, see for example \cite[Appendix 10]{Cohen2015}).
\end{proof}

\begin{theorem}
The value function of the control problem satisfies the recursion
\[\begin{split}
&V_{t-1} = \Xi_{t-1}(\kappa_{t-1}, \{Y_s\}_{s\le t-1})\\
&=\inf_{u \in U} \bigg\{\mathfrak{L}_{t}(u)+\sup_{p\in S_N^+, \fA\in \bA}\Big\{\int_{\bR^d} \Xi_t\big(\kappa^u_t(\cdot|y, \kappa_{t-1}), y,  \{Y_s\}_{s\le t-1}\big) c_\fA(dy; A^\fA p_{t-1})\\
&\qquad\qquad\qquad\qquad-\big(k^{-1}\kappa_{t-1}(p)\big)^{k'}\Big\}\bigg\}
\end{split}\]
where 
\[\begin{split}\kappa_t^u(\cdot|y, \kappa_{t-1}) = \inf_{\fA_t}\bigg\{\inf_{p_{t-1}\in (\cdot)^{\ofA_t}}\Big\{&\kappa_{t-1}(p_{t-1}) + \gamma_t(\fA_t; \{Y_s\}_{s<t}, u_t)\\ &-\log\Big(c_\fA (y; A^\fA_t p_{t-1}, u_t)\Big)\Big\}\bigg\}-m_t\end{split}\]
for $m_t$ a normalizing constant (which may depend on $u_t$) to ensure $\kappa_t^u$ has minimal value zero, and terminal value
\[\Xi_{T}(\kappa_{T}, \{Y_s\}_{s\le T}) := \inf_{p\in S_N^+} \Big\{\sum_i \mathfrak{C}(e_i, \{Y_s\}_{s\le T}) p_i - \big(k^{-1} \kappa_T(p)\big)^{k'}\Big\}.\]
\end{theorem}

\begin{proof}
We proceed by recursion, and assume that an optimal policy is always attainable (this can be relaxed, with an increase of notational complexity). We suppose that the value function is given by $V_t=\Xi_t(\kappa_t^{u^*}, \{Y_s\}_{s\le t})$, for some function $\Xi_t:\mathfrak{K}\times \bR^t\to \bR$, where $\mathfrak{K}$ denotes the space of functions from $S_N^+$ to $\bR$ with minimal value zero. From our assumptions, this is clearly true at time $T$, and $\Xi$ has the stated form.
 
 From the perspective of time $T-1$, suppose the current penalty state $\kappa_{T-1}$ is given. Then, for any proposed control $u_T$ (which is assumed to be $\CY_{T-1}$-measurable), when $Y_T$ is observed at time $T$, the time-$T$ penalty will be given by $\kappa_T^u(\cdot|Y_T, \kappa_{T-1})$. For each choice of $\kappa_{T-1}$ and each conditionally Markov measure $\bQ$ (with generator $\fA$), the remaining cost faced at time $T$ will be 
\[\Xi_T\big(\kappa^u_T(\cdot|Y_T, \kappa_{T-1}), Y_T, \{Y_s\}_{s\le T-1}\big).\] 
The conditional expectation of $\Xi_T$ can be written 
 \[\begin{split}
 \bE_\bQ[\Xi_T(\kappa^u_T)|\CY_{T-1}] &= \int_{\bR^d} \Xi_T(\kappa^u_T(\cdot|y, \kappa_{T-1}), y, \{Y_s\}_{s\le T-1}) c_\fA(dy; A^\fA p_{T-1}).
 \end{split}\]
  
 The DR-expected cost of using control $u_T$, as seen from time $T-1$, fixing the state of the penalty at $T-1$, is then given by taking a conditional expectation and adding the running cost term, which yields
 \[\begin{split}&
 \mathfrak{L}_{T-1}(u_T) + \oCE_u\big(\Xi(\kappa^u_T, \{Y_s\}_{s\le T})\big|\CY_{T-1}\big)\\
&= \mathfrak{L}_{T-1}(u_T) \\
&\qquad+\sup_{p\in S_N^+, \fA\in \bA}\Big\{ \int_{\bR^d} \Xi_T\big(\kappa^u_T(\cdot|y, \kappa_{T-1})\big) c_\fA(dy; A^\fA_{T-1} p)-\big(k^{-1}\kappa_{T-1}(p)\big)\Big\}.
 \end{split}\]
Finally, optimizing this cost gives the one-step minimal cost in terms of a function of $\{Y_s\}_{s\le T-1}$ and $\kappa_{T-1}$, as required. (To be rigorous, this can be done using a measurable selection argument, as in \cite[Appendix 10]{Cohen2015}.) The result follows by recursion, as our problem is known to satisfy the dynamic programming principle.
\end{proof}

\begin{corollary}
A control is optimal if and only if it achieves the infimum in the formula for $V_t$ above. 
\end{corollary}

\begin{remark}
If we assume that the terminal cost depends only on $X_T$ (and not on $Y$), and the running cost does not depend on $Y$, then one can observe a Markov property to the control problem, that is, $V_s$ is conditionally independent of $Y$ given $\kappa_s$. The corresponding optimal controls can then also be taken only to depend on $\kappa_s$.
\end{remark}

\begin{remark}
We can write (using $\partial$ for the partial difference operator and omitting dependence on $Y$)
\[\partial_t \Xi_t (\kappa) = \Xi_t(\kappa) - \Xi_{t-1}(\kappa), \qquad \partial_{\kappa}\Xi_t(\kappa') = \Xi_{t}(\kappa+\kappa')-\Xi_{t}(\kappa)\]
Rearranging, we obtain the following Bellman--Isaacs-type equation for $\Xi$
\[\begin{split}
 \partial_t \Xi_t(\kappa)=-\inf_{u \in U} \sup_{p\in S_N^+, \fA\in \bA}\Big\{&\int_{\bR^d}   \partial_{\kappa}\Xi_t\big(\kappa^u_t(\cdot|y, \kappa)-\kappa\big)c_\fA(dy; A^\fA p)\\
&+\mathfrak{L}(u_T) -\big(k^{-1}\kappa_{T-1}(p)\big)^{k'}\Big\},
\end{split}\]
with terminal value $\Xi_T(\kappa) = \inf_{p\in S_N^+}\big\{\sum_i\mathfrak{C}(e_i)p_i-\big(k^{-1}\kappa(p)\big)^{k'}\big\}.$
\end{remark}

\bibliographystyle{plain}
\bibliography{PaperBib}

\end{document}